\documentclass[11pt]{article}
\usepackage[a4paper,margin=2.5cm]{geometry}
\usepackage[utf8]{inputenc}
\usepackage[T1]{fontenc}
\usepackage[linesnumbered,ruled,vlined]{algorithm2e}
\usepackage{siunitx}
\usepackage{amsmath}
\usepackage{amsthm}
\usepackage{amssymb}
\usepackage{subcaption}
\usepackage{graphicx}
\usepackage{caption}
\usepackage{authblk}
\usepackage{hyperref}
\usepackage[sort]{cite}
\usepackage{tikz}
\usepackage{comment}

\usetikzlibrary{arrows.meta,positioning,calc,fit}
\definecolor{granat}{RGB}{0,55,103}
\definecolor{oi-blue}{HTML}{0072B2}
\definecolor{oi-orange}{HTML}{E69F00}
\definecolor{oi-green}{HTML}{009E73}
\definecolor{oi-yellow}{HTML}{F0E442}
\definecolor{oi-pink}{HTML}{CC79A7}
\definecolor{oi-lightblue}{HTML}{56B4E9}
\definecolor{oi-red}{HTML}{D55E00}
\definecolor{oi-purple}{HTML}{882255} 
\newtheorem{theorem}{Theorem}

\newtheorem{corollary}[theorem]{Corollary}

\theoremstyle{definition}

\title{A Graph Theoretic Approach to Spatial Modeling of Post Disaster Shelter Camps Using Rainbow and Roman Domination Parameters}

\author[1]{Didem Gözüpek\thanks{Corresponding author: \href{mailto:didem.gozupek@gtu.edu.tr}{didem.gozupek@gtu.edu.tr}}}
\author[2]{Piotr Lange\thanks{\href{mailto:piotr.lange@student.pg.edu.pl}{piotr.lange@student.pg.edu.pl}}}
\author[2]{Joanna Raczek\thanks{\href{mailto:joanna.raczek@pg.edu.pl}{joanna.raczek@pg.edu.pl}}}

\affil[1]{Computer Engineering Department, Gebze Technical University, 41400 Gebze, Kocaeli, Turkey
}
\affil[2]{Faculty of Electronics, Telecommunications and Informatics, Gdańsk University of Technology, Narutowicza 11/12, 80-233 Gdańsk, Poland
}

\begin{document}

\maketitle
\begin{abstract}
Effective spatial organization of post-disaster shelter camps is essential for ensuring access to basic services while making efficient use of limited space and resources. In this paper, we propose a graph-theoretic framework for shelter-camp facility placement based on rainbow $k$-domination and Roman domination. Rainbow $k$-domination models the simultaneous accessibility of distinct facility types, such as sanitation units, kitchens, water points, and schools, whereas Roman domination is used to represent services with different capacity levels, illustrated through Wi-Fi deployment. We present an $O(nk)$-time algorithm for finding a minimum rainbow $k$-dominating set of a tree with $n$ vertices, which is linear in $n$ for fixed $k$, and computational experiments confirm its scalability on large instances. For general graphs, we establish new lower and upper bounds on the rainbow $k$-domination number. We further investigate its relationship with Roman domination, derive structural properties of graphs attaining the extremal equality between the two parameters, and prove that recognizing such graphs is NP-hard. These results provide a theoretical foundation for domination-based approaches to facility placement in post-disaster shelter planning.
~~\\

\end{abstract}

%\flushbottom
%\thispagestyle{empty}

\section{Introduction}

In contemporary settings, emergencies can necessitate rapid and coordinated responses, including the temporary relocation of populations to ensure their safety. Such circumstances may result from natural disasters, technological failures, geopolitical instability, or other disruptions that pose immediate risks to human life. The effective management of these events depends not only on the timely mobilization of resources but also on the ability to establish safe and well-organized temporary settlements that protect affected individuals while broader stabilization measures are implemented. Consequently, preparedness for both short and long-term population displacement has become a critical component of modern crisis response frameworks.

Although natural and anthropogenic hazards have long been present, the scale of their humanitarian impact has increased, largely due to rising population density and the expansion of settlements into vulnerable regions. In earlier periods, many disasters affected relatively few people, as extensive areas were sparsely inhabited or unpopulated. Today, however, with significantly more individuals residing in exposed coastal zones, seismic regions, floodplains, and densely urbanized areas, events of comparable magnitude can result in substantially greater displacement and far more complex sheltering demands.

Within this context, shelter allocation constitutes a critical component of comprehensive disaster management frameworks, as it directly influences humanitarian outcomes and overall system resilience. Existing research on shelters in disaster management focus primarily on determining shelter locations and/or assigning survivors to shelters, while largely neglecting the design and layout of shelter sites (see \cite{su11020399, kilci2015locating, li2017hierarchical, eriskin2024applying, chen2013temporal, zhao2015scenario}). To the best of our
knowledge, the only study explicitly addressing the optimization of shelter layout design is \cite{karsu2019refugee}, where the
authors employ a general block design structure to account for facility requirements and their interrelations.

\subsection{Problem Description and Motivation}
In the aftermath of large-scale disasters, affected populations often require temporary living spaces for extended periods. As a result, emergency shelter sites frequently evolve into structures resembling refugee camps, which demand not only substantial physical space but also careful and deliberate spatial organization.

Humanitarian shelter camps commonly rely on prefabricated container units, referred to as camp containers, to provide temporary housing, as well as essential services such as health clinics, administrative offices, and sanitation facilities. These units are typically purpose-built modular structures rather than repurposed shipping containers. According to available sources, standard units generally have widths of approximately 2.4--3.0 m, lengths of 6--7 m, and heights around 2.6--2.8 m \cite{prefabrik2022container}. Depending on functional needs, variations in size can be manufactured, and modular designs enable multiple units to be combined into larger, more complex structures.

Given that such shelter sites are often intended to operate for prolonged durations and may accommodate thousands of individuals \cite{refugeeUNHCR2021, Lanati2023}, ensuring adequate living conditions becomes a critical concern. In response, the United Nations Refugee Agency has established comprehensive planning standards to guide the design and organization of these environments \cite{unhcr2026emergencyhandbook}. Table~\ref{table:unhcr} summarizes selected requirements from these standards. Notably, service facilities such as communal latrines and water points are required to be located within specified distances of residential units to ensure accessibility, safety, and hygiene. Similar proximity considerations can also be extended to other essential facilities, such as healthcare units, in order to further enhance the overall quality of life within the camp.

\renewcommand{\arraystretch}{1.3}
\begin{table}[h!]
\centering
\small
\begin{tabular}{>{\raggedright\arraybackslash}p{3.2cm} 
                >{\centering\arraybackslash}p{3.5cm} 
                >{\raggedright\arraybackslash}p{6.5cm}}
    \hline
    \textbf{Description} & \textbf{Standard} & \textbf{Further consideration} \\
    \hline
    Communal latrine (toilets) & 1 per 20 persons & Separate areas for men and women. For long-term accommodation, use one household latrine per family. \\
    Latrine distance & 6--50\,m from shelter & Must be close enough to encourage use but far enough to prevent smells and pests. \\
    Water tap stands & 1 per 80 persons & 1 per community. \\
    Water distance & Max.\ 200\,m from household & No dwelling should be further than a few minutes' walk from a distribution point. \\
    Health center & 1 per 20{,}000 persons & 1 per settlement; include water and sanitation facilities. \\
    School & 1 per 5{,}000 persons & 3 classrooms, 50\,m\textsuperscript{2}. \\
    \hline
\end{tabular}
\caption{\label{tab:distances}Practical standards for container-based shelter camps based on the United Nations Refugee Agency Camp Planning Standards \cite{unhcr2026emergencyhandbook}.} \label{table:unhcr}
\end{table}

We assume that an initial spatial layout of containers is given. In practice, containers are often arranged in a grid-like structure, with typical distances of approximately 10~m between adjacent units and 15--30~m between sectors. However, the proposed approach is not restricted to grid layouts and can accommodate arbitrary spatial configurations.

To capture service accessibility, we model the layout as a graph in which each container corresponds to a vertex. Two vertices \(u\) and \(v\) are connected by an edge if the distance between \(u\) and \(v\) is within a certain range. For simplicity, we assume a uniform service range across all facility types, for instance, between 20 m and 40 m, so that an edge is included whenever the distance between the corresponding containers lies within this interval. This graph representation encodes the feasible service relationships between containers.

Given this graph, the goal is to determine which containers should be designated as facilities, their corresponding facility types, as well as which should remain residential. To this end, we adopt a graph-based labeling framework in which each facility type is represented by a distinct label, and each facility container is assigned exactly one label. Residential containers correspond to unlabeled vertices; however, they must have access to all required facility types within their neighborhood. Thus, every residential container must be ``covered'' by nearby facility containers that collectively provide the full set of essential services (e.g., water, sanitation, food distribution). This requirement naturally leads to a formulation based on (singleton) rainbow $k$-domination, defined as follows.

Consider a graph $G$ together with a collection of $k$ colors, where each vertex is assigned any subset of these colors. Such an assignment is called a $k$-rainbow dominating function if every vertex receiving no color has neighbors whose assigned colors collectively include all $k$ colors. The associated parameter $\gamma_{rk}(G)$ is defined as the minimum total number of colors assigned across all vertices \cite{brevsar2008rainbow}.

If we further restrict the assignment so that each vertex receives at most one color, then the concept is known as \emph{rainbow $k$-domination} \cite{bonomo2018domination} (or \emph{singleton rainbow domination}) \cite{ervevs20213}, and the corresponding parameter is denoted by $\tilde\gamma_k(G)$. In a shelter camp, each container is a single-floor unit that can serve as either one type of facility or a residential space. This corresponds to the constraint that at most one color may be assigned to each vertex. Consequently, the practical problem in this setting directly maps to the singleton rainbow domination problem. This problem has also recently appeared in the literature as the \emph{dual server domination problem} \cite{CHELLALI2026251}, corresponding to the case $k=2$, where the authors provide a structural characterization of trees whose rainbow $2$-domination number equals $n-2$.

\begin{figure*}[t]
    \centering
    \includegraphics[width=\textwidth]{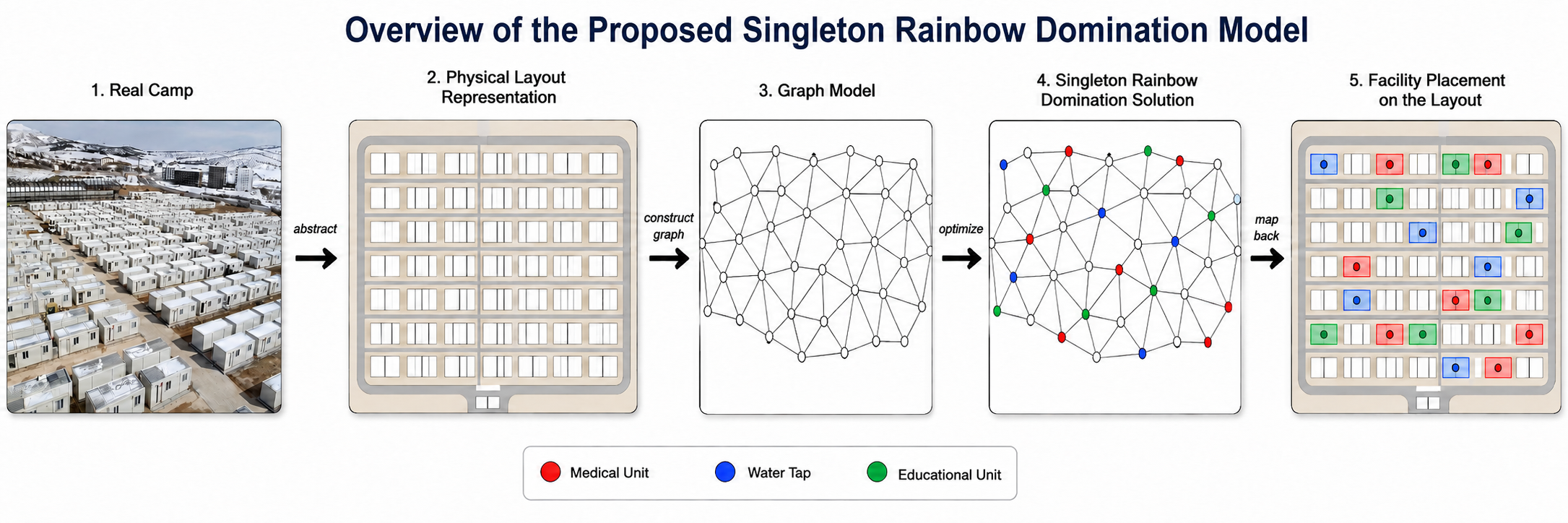}
    \caption{Overview of the proposed graph-based facility placement model for a post-disaster container camp using singleton rainbow domination.}
    \label{fig:proposed-singletonrainbow-model-overview}
\end{figure*}

\begin{figure*}[t]
    \centering
    \includegraphics[width=\textwidth]{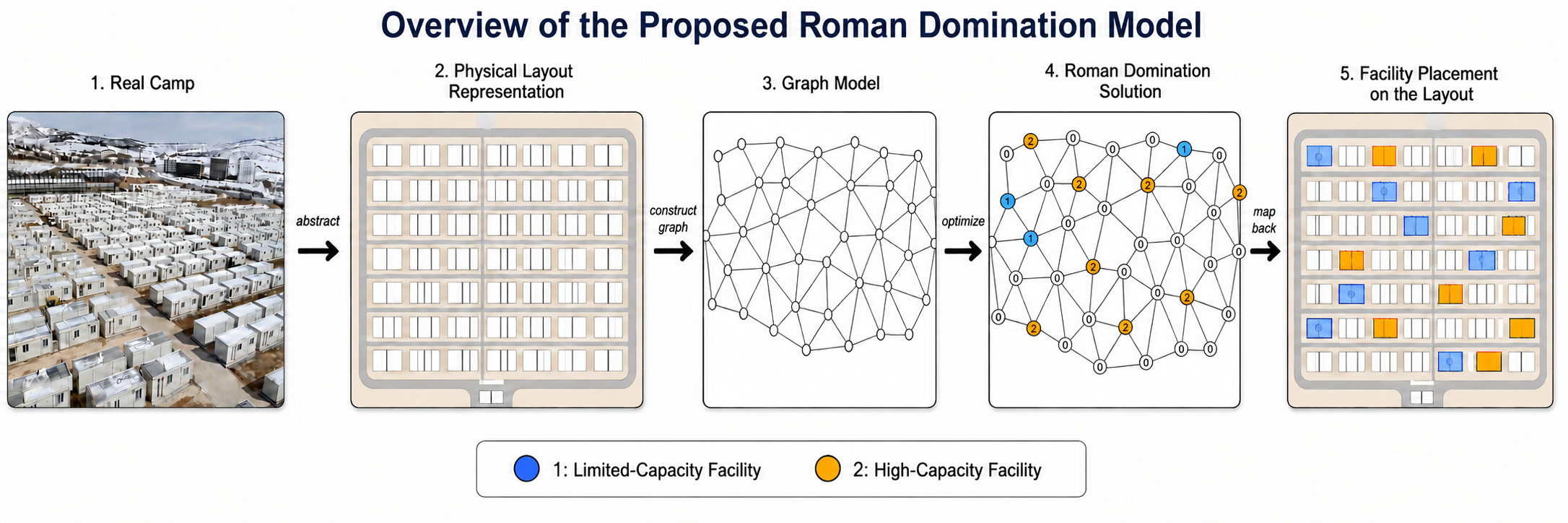}
    \caption{Overview of the proposed graph-based facility placement model for a post-disaster container camp using Roman domination.}
    \label{fig:proposed-singletonrainbow-model-overview2}
\end{figure*}

Figure~\ref{fig:proposed-singletonrainbow-model-overview} illustrates an example of a shelter-camp and its corresponding graph model, where vertices represent containers and edges indicate admissible proximity for service provision. Colored vertices correspond to facility containers assigned different service types such as medical units, water taps, and educational units, while uncolored vertices represent residential containers. Each residential container must be adjacent to facility containers covering all required service types, corresponding to singleton rainbow domination. The output of the singleton rainbow domination model can then be mapped back to the original facility layout to determine the locations of the required facilities.

Providing wireless Internet access is an important service in modern shelter camps. Since such facilities are usually established under budget constraints, communication infrastructures should be inexpensive, easy to deploy, and energy-efficient. While modern Wi-Fi networks often rely on 5~GHz devices offering higher throughput, the corresponding coverage area is typically smaller and signal attenuation caused by obstacles is more severe. In contrast, 2.4~GHz access points provide better coverage and penetration through obstacles, making them attractive for temporary and low-cost deployments. Typical Wi-Fi access points operating in the 2.4 GHz band provide  outdoor coverage reaching 90-150~m under favorable conditions and indoor coverage of approximately 30-45~m, because different materials attenuate wireless signals to varying degrees, so the actual coverage can vary significantly depending on the wall composition and its thickness. For this reason, waterproof outdoor-rated access points, although more expensive, can be installed outside the shelter, significantly increasing Wi‑Fi coverage and reducing signal attenuation caused by the walls.

 This requirement naturally leads to a formulation based on Roman domination, defined as follows. Given a graph $G$, Roman domination considers an assignment of the values $0$, $1$, and $2$ to its vertices. The defining condition is that every vertex assigned the value $0$ must have at least one neighbor with value $2$. The weight of such an assignment is the sum of the values assigned to all vertices, and the Roman domination number $\gamma_R(G)$ is the minimum possible weight over all assignments satisfying this condition \cite{cockayne2004roman}.
Since, in our model, an edge joins two shelters located within a distance of 40 meters, Roman domination provides a natural framework for determining a cost-effective deployment of Wi-Fi routers throughout a shelter camp. A vertex assigned value 1 may correspond to a standard access point, potentially implemented using older and less expensive hardware with limited throughput. In contrast, a vertex assigned value 2 may represent a more capable networking device offering higher throughput, larger coverage, or additional communication resources. The resulting Roman dominating function models a cost-efficient deployment of wireless infrastructure while ensuring that every container without a router is located within the service area of at least one high-capacity communication node.

%In the context of shelter facility layout, the values $0$, $1$, and $2$ may be interpreted as different facility-capacity levels. A vertex assigned $0$, $1$ and $2$ corresponds to a residential container, a limited-capacity facility, and a high-capacity facility, respectively. For instance, a vertex assigned color $1$ can represent a small first-aid point, a small sanitation unit, or a basic supply point, whereas a vertex assigned color $2$ can represent a larger medical unit, a central sanitation facility, or a main distribution center. The role of value $1$ is to model situations in which a local facility is needed, but installing a high-capacity facility would be unnecessary or too costly.

Figure \ref{fig:proposed-singletonrainbow-model-overview2} illustrates the corresponding scenario for Roman domination, where vertices with value 1 (blue vertices) and 2 (orange vertices) correspond to low and high capacity facilities, respectively. Uncolored vertices, equivalently vertices assigned value (0), must be located within the service range of at least one high-capacity facility. As in the singleton rainbow domination setting, the resulting vertex assignments can subsequently be mapped back to the physical layout to determine the locations and capacity levels of the facilities.

Rainbow domination and its variants have been extensively studied in graph theory as natural generalizations of classical domination concepts, capturing scenarios where multiple types of resources or services must be simultaneously provided \cite{brevsar2005paired, Bresar2020}. These parameters arise in applications such as frequency assignment, sensor networks, and resource allocation problems, where different functionalities must be distributed across a network while ensuring full coverage. Roman domination, on the other hand, has attracted significant attention due to its strong combinatorial structure and connections to defense and protection models \cite{cockayne2004roman, Chellali2020}. Despite their rich theoretical development, these parameters have not yet been systematically explored in the context of spatial planning for humanitarian shelter systems, where multiple essential services must be jointly accessible under strict proximity constraints.

Existing studies on shelter allocation primarily focus on facility location models that optimize distances, capacities, or evacuation flows. However, these approaches typically treat different facility types independently or aggregate them into a single objective, thereby overlooking the combinatorial structure arising from the need to simultaneously provide multiple distinct services within constrained spatial neighborhoods.

To the best of our knowledge, there is no existing framework that models shelter-camp layout using graph-theoretic domination parameters capable of capturing multi-type service coverage in a unified and structurally analyzable way. In particular, the relationship between (singleton) rainbow 2-domination and Roman domination has not been investigated in this application domain. Addressing this gap is important for developing mathematically grounded and computationally efficient tools that can support the design of shelter layouts under realistic operational constraints.

The article first investigates the properties of the rainbow domination number and then explores its relationship with the Roman domination number. By analyzing the interplay between rainbow domination and Roman domination, we aim to identify structural dependencies that may simplify optimization procedures, as insights gained from one parameter can potentially facilitate the analysis and optimization of the other. 

The remainder of the paper is organized as follows. Section~2 presents a linear‑time algorithm for trees for rainbow $k$-domination together with empirical evaluation of this algorithm on large tree instances. Next section presents structural results for rainbow $k$-domination and relating this number to Roman domination number. This section finishes with proving NP‑hardness of deciding whether the rainbow 2-domination number is equal to $1.5$ times the Roman domination number.

\section{A Linear-Time Algorithm for Minimum Rainbow $k$-Dominating Set for Trees}
In this section, we introduce the definitions and fundamental properties of the singleton rainbow domination and Roman domination parameters in graphs, including several basic results from the literature. We then present a linear-time algorithm for finding a minimum singleton rainbow dominating set in trees.

%\subsection{Problem formulation}
  Let $G=(V,E)$ be a graph and let $k$ be a positive integer. A {\em rainbow $k$-dominating function} $f$ of $G$ assigns to each node of $G$ exactly one set of $\emptyset, \{1\},\{2\},\dots,\{k\}$ in such a way that if $f(v) = \emptyset$ for some node $v\in V(G)$, then $v$ is adjacent to a node $u$ with $f(u)=\{i\}$ for each $i=1,2,\dots, k$. In other words, if $f(v) = \emptyset$, then $|\bigcup_{u\in N(v)}f(u)| \ge k$. We will refer to elements of $\{\{1\},\{2\},\dots,\{k\}\}$ as {\em color sets}. 
  
  The {\em rainbow $k$-domination number} of $G$ is denoted by $\tilde\gamma_k(G)$ and equals the minimum number $\Sigma_{v\in V}|f(v)|$ over all rainbow $k$-dominating functions $f$ of $G$. Additionally, denote by $V_{\{i\}}$ the set of all nodes $v\in V$ with $f(v)=\{i\}$, where $i=1,2,\dots,k$ and let $V_\emptyset$ be the set of all nodes $v$ with $f(v)=\emptyset$.
  
For example, Fig.~\ref{Ex2} depicts a tree $T$ with
$\tilde{\gamma}_3(T)$=~11. Each color set is represented by a distinct vertex color, while white denotes the empty set. A graph may admit more than one minimum rainbow $k$-dominating set. For instance, either of the vertices $12$ and $14$, which belong to the red color class in Fig.~\ref{Ex2}, may instead be assigned to the green or blue color class without violating the rainbow $3$-domination condition. Moreover, for any distinct $i,j\in \{1,2,\ldots,k\}$, interchanging the color classes $V_{\{i\}}$ and $V_{\{j\}}$ yields another minimum rainbow $k$-dominating set of $G$.

%\subsection{A Linear-Time Algorithm for Trees}
Paper~\cite{CHANG20108} shows that determining the $k$-rainbow domination number is NP-hard even when restricted to bipartite graphs or to chordal graphs. The same paper presents a linear-time algorithm that finds minimum $k$-rainbow dominating sets in trees. Since this algorithm cannot be easily adapted for finding rainbow $k$-domination number, we present an algorithm that uses different approach that computes a minimum rainbow $k$-dominating set of a tree in $O(nk)$ time, where $n$ denotes the order of the tree and $k$ is the number of colors. In particular, the algorithm runs in linear time with respect to $n$ when $k$ is fixed. 

The tree-order is a partial ordering on the nodes of a rooted tree $T$, where $u < v$ if and only if the unique path from the root to $v$ passes through $u$. If $v\in V(T)$, then let $C(v)$ be the set of all children of a vertex $v$ (if $v$ is a leaf, then $C(v)=\emptyset$), and let $F(v)$ be the father of $v$ (if $v$ is the root, then $F(v)=\emptyset$). 

The algorithm, whose pseudocode is outlined in listing~\ref{RanibowF1}, proceeds the nodes of $T$ along its tree ordering twice: Phase~1 proceeds from the last to the root, while Phase~2 from the root to the last node.  Thus, the algorithm traverses the vertex set of $T$ twice.

\subsubsection*{Conceptual Outline of the Algorithm for $k=3$}
In Phase~1 each node is assigned exactly one set of $\emptyset, \{1\},\{2\},\dots,\{k\}$. This assignment is preliminary and may be refined during Phase~2. Each node gets an optimal color set according to the information passed to it by its descendants. 

\begin{figure}[h]
\begin{subfigure}[t]{0.49\textwidth}
\centering
\includegraphics[scale=0.62]{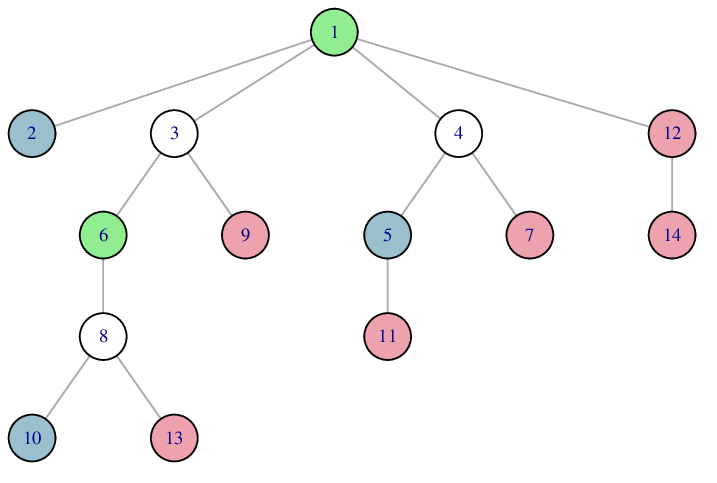}
\caption{Tree $T_1$ after performing Phase~1 of the algorithm for $k=3$}\label{Ex1}
\end{subfigure}
\begin{subfigure}[t]{0.49\textwidth}
\centering
\includegraphics[scale=0.62]{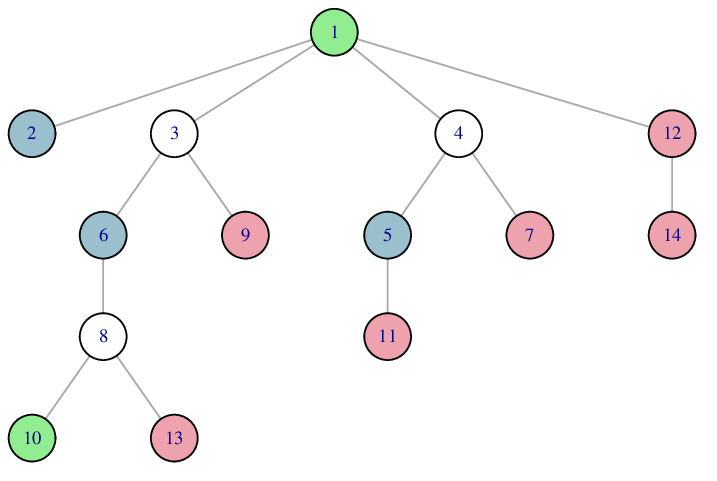}
\caption{Tree $T_1$ after performing Phase~2 of the algorithm for $k=3$}\label{Ex2}
\end{subfigure}
\caption{Linear algorithm for trees, situations~\ref{s1} and~\ref{s2}. Each color set is visualized using a unique vertex color, and white corresponds to the empty set.}
\end{figure}

For example, Fig.~\ref{Ex1}  presents a tree $T_1$ after performing Phase~1. This result is not valid, as $f(3) =\emptyset$, even though this vertex has access to only two colors in its neighborhood (rather than three). The assignment will be revised in Phase~2 of the algorithm (see Fig~\ref{Ex2}). 

An analogous situation arises in tree $T_2$ (see Fig.~\ref{Ex3}, where vertex~8 is not adjacent to three distinct color sets. However, after Phase~2 (illustrated in Fig.~\ref{Ex4}), a minimum rainbow‑3 dominating set is obtained.

The example trees $T_1$ and $T_2$ illustrate how the algorithm operates and what types of situations may arise during its execution. The four basic scenarios are listed and explained below.
\begin{enumerate}
    \item \label{s1} When determining which color set to assign to a vertex $v$ which has at most $k-2$ distinct color sets in $C(v)$, the algorithm first checks which color sets have not yet been used by its already-processed siblings. It then assigns, whenever possible, one of these unused color sets to $v$. This situation happens when node~10 is processed.
    \item \label{s2}If a vertex $v$ has exactly $k-1$ distinct color sets among its children and $v$ is not the root, then we set $f(v)=\emptyset$, and the parent of $v$ is required to receive the missing color set while being processed. This situation takes place when nodes~4 and~8 are processed. Since vertex~3 is examined by the algorithm after vertex~4, its parent, vertex~1, cannot be forced to take a different color. In this situation, the information about recoloring green vertices to blue and blue vertices to green in the sub-tree rooted in~3 is stored at vertex~3 and propagated to its descendants during Phase~2.
    \item \label{s3} If a vertex $v$ contains all color sets in $C(v)$, then $v$ may be assigned $f(v)=\emptyset$, as illustrated for vertex~1 in Fig.~\ref{Ex3}.
    \item \label{s4} If a vertex~$v$ is required by one (or more) of its children to receive color set $\{j\}$, then $f(v)=\{j\}$ providing this color lacks among children of $F(v)$. For example, vertex~6 (see Fig.~\ref{Ex3}) demands vertex~4 to be green. Green is missing at vertex~1, so vertex~4 is assigned color set green. Notwithstanding identical situation for vertex~2, this node gets color blue, since it is processed later than~4 and blue lacks in the neighborhood of vertex~1. Hence, all descendants of 2 interchange colors green  and blue during phase~2, which results in an optimal solution (see Fig.~\ref{Ex4}).
\end{enumerate}
\begin{figure}[h]
\begin{subfigure}[t]{0.49\textwidth}
\centering
\includegraphics[scale=0.62]{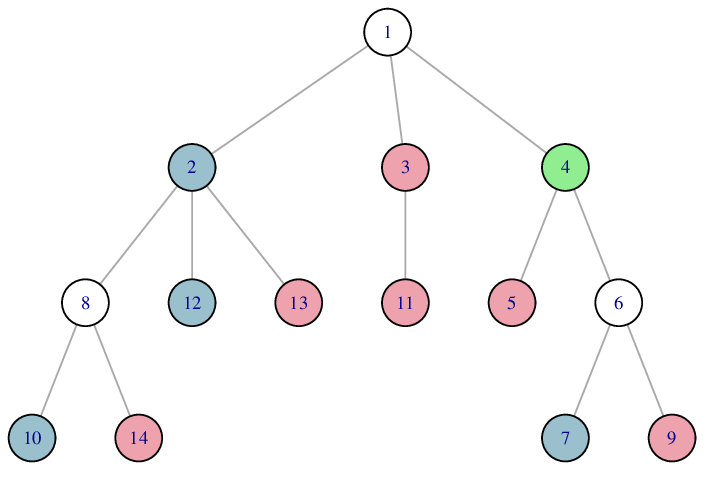}
\caption{Tree $T_2$ after performing Phase~1 of the algorithm for $k=3$}\label{Ex3}
\end{subfigure}
\begin{subfigure}[t]{0.49\textwidth}
\centering
\includegraphics[scale=0.62]{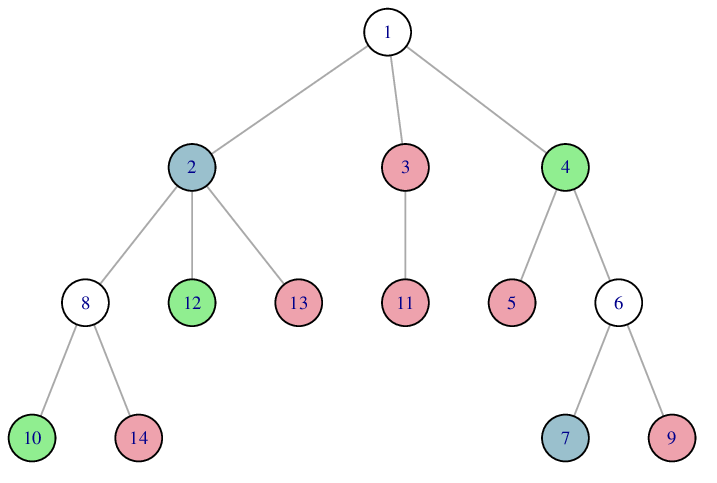}
\caption{Tree $T_2$ after performing Phase~2 of the algorithm for $k=3$}\label{Ex4}
\end{subfigure}
\caption{Linear algorithm for trees, situations~\ref{s3} and~\ref{s4}. Each color set is visualized using a unique vertex color, and white corresponds to the empty set.}
\end{figure}

Each node $v$ stores the following variables:
\begin{itemize}
    \item A vector $v.N$ of size $k$. If $v.N[i]>0$, then at least one child of $v$ is of color $i$. It is initialized with zeroes.
    \item A color set $\{c\}$ assigned to a node is stored as an integer in $v.c$. If $v.c=0$, then $v$ is assigned an empty set.
    \item Number $v.r$ stores the integer representation of a color set that is missing in the neighborhood of a child (or children) with an empty set assigned, and therefore $v$ should be assigned the color set $\{r\}$. 
    \item A vector $v.\pi$ of size $k$ which represents a permutation. It is initialized with consecutive values from ~1 to $k$, that is the identity permutation. 
    %We use the standard right-to-left convention for composition of permutations. Thus, for permutations $v.\pi$ and $u.\pi$, $v.\pi \circ u.\pi = v.\pi(u.\pi)$.
\end{itemize}
We assume that at the start of the algorithm for each node $v$, $v.r=0$ and $v.N$ is a zero vector.
We also assume that if $F(v)=null$, then all conditions in the {\tt if} statements evaluate to false and all statements containing $F(v)$ are not evaluated.

The proof of the correctness of the algorithm finding a minimum rainbow $k$-dominating set is based on the idea of dynamic programming and mathematical induction. We prove that after processing each vertex $v$ of the input tree the sub-tree root at this node has a valid and optimal assignment of color providing recoloring is done in Phase~2 and $F(v)$ gets a demanded by $v$ color (situation~\ref{s2}).

\SetKwComment{tcp}{$\triangleright$\ }{}
\begin{algorithm}
\SetAlgoNoEnd
\SetAlgoVlined
\KwData{A tree $T$ with a tree-order}
\KwResult{$T$ with optimal assignment of colors to vertices}
\SetKwFunction{FMain}{{\tt Rainbow $k$}}
\SetKwProg{Fn}{Function}{:}{uiuyt}
    \Fn{\FMain{$T$}}{
\For(\tcp*[f]{Phase 1}){$v\in V(T)$ in reverse tree order}{
    \eIf{$v.r >0$}{
        \If{$F(v).N[v.r]>0$ and the $i$-th entry of $F(v).N$ is zero}{
            swap $v.\pi[v.r]$ and $v.\pi[i]$\;
            $v.r\gets i$\;
        }
        $v.c\gets v.r$\;
        $F(v).N[v.r]\gets F(v).N[v.r] + 1$\;
    }{
    \If{$v.N$ at least two zeroes, say $v.N[i]=0$}{
        \lIf{$F(v).N$ has a zero entry at position $j$}{$i\gets j$}
        $v.c\gets i$\;
        $F(v).N[i]\gets F(v).N[i] + 1$\;
        }
    \ElseIf{$v.N$ has exactly one zero entry, say $v.N[i]=0$}{
        \If{$F(v).r>0$ and $F(v).r\neq i$}{
            swap $v.\pi[F(v).r]$ and $v.\pi[i]$\;
            $i\gets F(v).r$\;
        }
    $F(v).r\gets i$\;
    }
    }
}%FOR
\lIf{$root.c=\emptyset$ and for some $i$, $root.N[i]=0$}{$root.c\gets i$}
\For(\tcp*[f]{Phase 2}){$v\in V(T)$ in tree order}{
    \lIf{$v.c \neq \emptyset$ and $v.c\neq F(v).\pi[v.c]$}{$v.c\gets F(v).\pi[v.c]$}
    $v.\pi\gets F(v).\pi(v.\pi)$\;
} %FOR
    \KwRet{$T$} 
 } % Main
 \caption{{\tt Rainbow $k$}-dominating function}\label{RanibowF1}
\end{algorithm}

\begin{theorem}\label{thmalgtrees}
    The algorithm {\tt Rainbow $k$-dominating function}, (see Algorithm~\ref{RanibowF1}), finds a minimum rainbow $k$-dominating set in a linear time.
\end{theorem}

\begin{proof}
    It is no problem to see that the data size grows linearly with the size of the input. Also, since the algorithm contains two independent {\tt for} loops that go through each node of the input tree once, the computation complexity is also linear. Therefore, it is enough to prove that the algorithm finds a minimum rainbow $k$-dominating set.

    If $v$ is a leaf, then $v.r=0$ and $v.N$ is a zero vector. Hence, lines 10--13 are performed. $v$ gets a color, preferably not used by other children of $F(v)$. In this way a one-node tree is optimally colored, which constitutes the base case.
    
    Denote by $T_v$ a subtree of $T$ rooted in $v$ containing $v$ and all its children. We maintain the invariant that, at any step of the main loop of Phase~1 of the algorithm, the subtree $T_v$ colored so far admits an optimal coloring if $F(v).r=0$, up to a potential recoloring performed in Phase~2, and if $F(v).r\neq 0$, then $T_v$ has an optimal coloring providing $F(v)$ gets color $F(v).r$. We prove that during performing the algorithm, the subtree expands and the invariant continues to hold.

    Let $v$ be a node currently processed by the main loop. Assume that the invariant holds for all children $u$ of $v$. Thus each subtree $T_u$ is optimally processed and each child $u$ stores a correct state. 

    If $v.r>0$, then at least one child of $v$ is missing color $v.r$ and lines 2--8 are performed. If color $v.r$ is already used by another child of $F(v)$ and there is a free color, colors are switched (lines 4--6), so that $v$ gets optimal color. Note that in case of recoloring, $v$ gets the target color and the information of recoloring of the nodes in $V(T_v)-\{v\}$ is stored in $v$. In Phase~2, by composition of permutations, this information is propagated down the subtree all the way to the leaves and, if needed, recoloring takes place. 

    Hence assume $v.r=0$. Then no color is required at $v$ by any of its children. Assume further that at least two color sets are lacking among the children of $v$. Then $v$ cannot be assigned an empty set and at least two entrances of $v.N$ are equal zero. In this case lines 10--13 are performed. To find an optimal color for $v$, the algorithm finds an index $i$ such that $F(v).N[i]=0$ and assigns color $i$ to $v$. If $F(v)$ does not have a zero entry, then $v$ gets any color. 

    Now assume that $v.r=0$ and $v.N$ is equal zero at exactly one entry, say $i$. Then $v$ may be assigned an empty set, providing its father gets the color set $\{i\}$. If $v$ is not the root, lines~14--18 are performed. If the father of $v$ already is required to be assigned a color set different than $\{i\}$ by another child, then the algorithm recolors $T_v$ (lines 15--17 and Phase~2), similarly as in the previous case. If $v$ is the root, then it gets a proper color set in line~19. 

    At last, if $v.r=0$ and $v.N$ does not have any zero entries, then $v$ is assigned an empty set, which is the default value for $v.c$. 

    In all cases, since the invariant holds for all children $u$ of $v$, we conclude that it is also true for $T_v$, which finishes the proof.
\end{proof}

\subsection*{Experimental evaluation }
To assess the practical applicability of the proposed algorithm, we developed a C++ implementation together with an automated testing and benchmarking suite. The implementation was used for both correctness verification and empirical runtime evaluation.

\subsubsection*{Correctness testing}
To verify the correctness of the implementation, we developed an automated test suite. The test procedure generates random rooted trees and runs the algorithm on those instances. For each computed solution, we verify that the resulting set is a valid singleton $k$-rainbow dominating set. Moreover, for small tree instances we compute the optimal solution using an exhaustive search procedure and compare it with the result produced by the algorithm.

The implementation successfully passed all tests, including thousands of randomly generated instances and all exhaustive verification cases for trees up to the size $n = 13$.

\subsubsection*{Runtime testing}
To empirically evaluate the time complexity, we conducted two sets of runtime experiments: one varying the number of vertices $n$ while keeping $k$ fixed, and another varying the number of colors $k$ while keeping $n$ fixed. 

In the first experiment, the algorithm was tested on randomly generated rooted trees of increasing order. In the second experiment, designed to assess the dependence on $k$, the algorithm was evaluated on rooted trees constructed to contain many high-degree vertices.

For each tested value of $n$ or $k$, the algorithm was executed multiple times and the median running time was recorded. In the first experiment (see Fig.~\ref{fig:runtime}), the number of vertices $n$ ranged from $5{,}000$ to $500{,}000$, while $k$ was kept fixed at $3$. In the second experiment (see Fig.~\ref{fig:runtime_k}), $n$ was fixed at $1{,}000{,}000$ and $k$ ranged from $1$ to $40$.

\begin{figure}[htbp]
\centering
\includegraphics[width=0.8\linewidth]{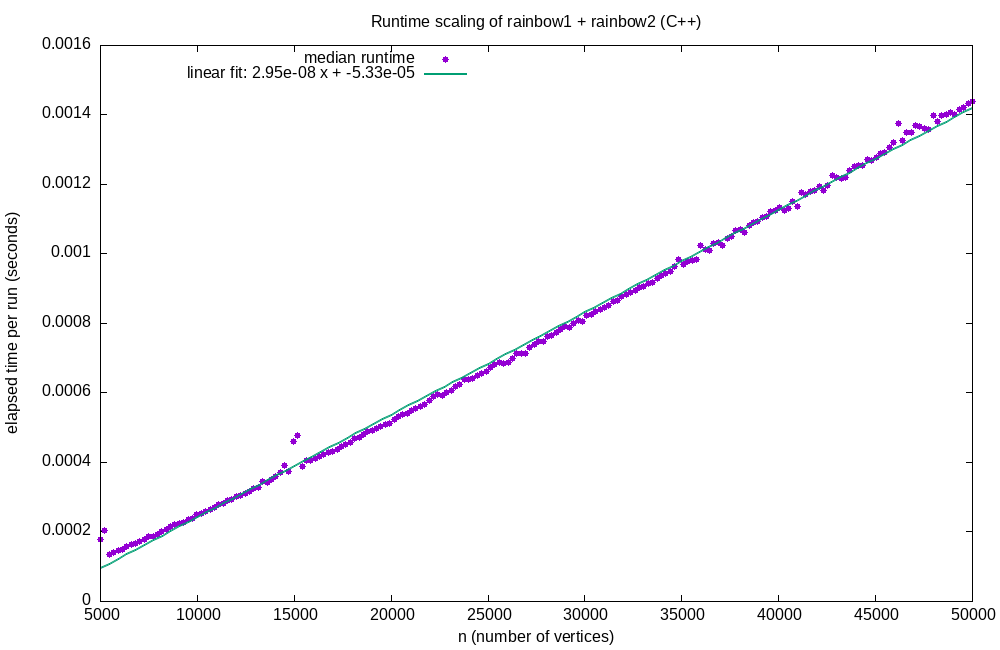}
\caption{Median running time of the C++ implementation of the proposed algorithm on randomly generated rooted trees of increasing order. The fitted line illustrates the approximately linear growth of the running time with respect to the number of vertices.}
\label{fig:runtime}
\end{figure}

\begin{figure}[htbp]
\centering
\includegraphics[width=0.8\linewidth]{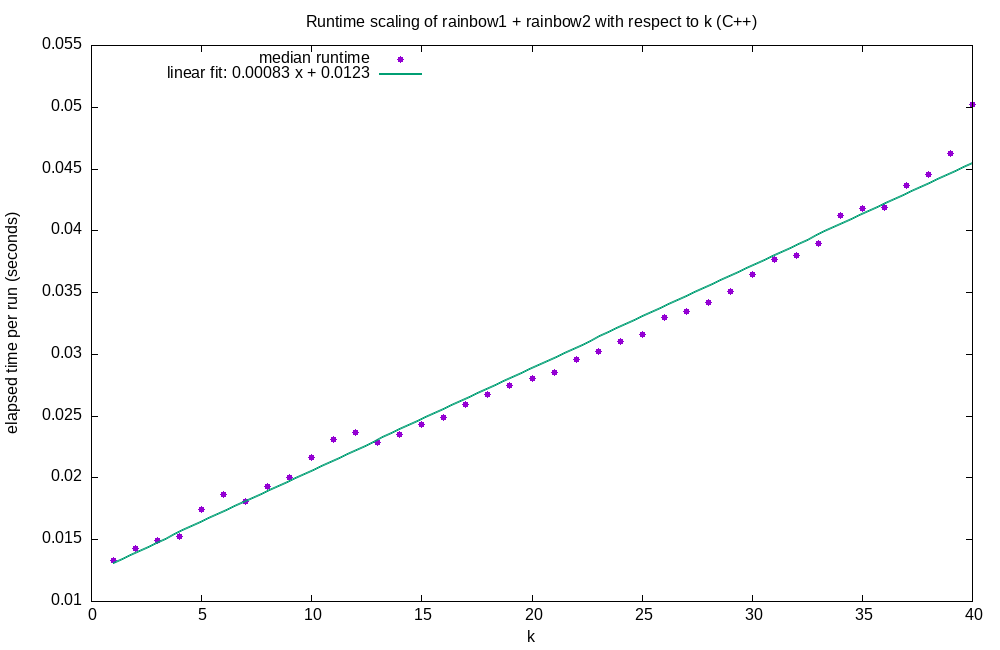}
\caption{Median running time of the C++ implementation as a function of the number of colors $k$ for a fixed number of vertices. The fitted line indicates approximately linear growth with respect to $k$.}
\label{fig:runtime_k}
\end{figure}

The results presented in Fig.~\ref{fig:runtime} confirm that the running time grows approximately linearly with the number of vertices. Similarly, Fig.~\ref{fig:runtime_k} shows that the running time scales approximately linearly with respect to the number of colors $k$. These observations are consistent with the theoretical time complexity of the algorithm, which is linear in both $n$ and $k$. Minor deviations from the linear trend can be observed in the measurements. These variations are typical in empirical runtime experiments and are primarily caused by factors unrelated to the algorithm itself, such as operating system scheduling, processor frequency scaling, cache behavior, and other hardware-level effects.

%\subsection*{Details of experimental setup}

\subsubsection*{Implementation}
The algorithm was implemented in C++. The input tree is represented by a parent array, where for each vertex $v$ the value $parent[v]$ stores the father of $v$ in the rooted tree. This representation provides constant-time access to the parent of each vertex and allows both phases of the algorithm to be performed by simple linear scans of the vertex set.

For each vertex, the implementation stores:
(i) its assigned color,
(ii) a temporary requirement value propagated from descendants,
(iii) a vector recording which colors are already represented in its processed neighborhood, and
(iv) a recoloring permutation used during the second phase of the algorithm.

The implementation used in the experiments corresponds directly to the pseudocode presented in Algorithm~\ref{RanibowF1}.

\subsubsection*{Correctness verification procedure}
Correctness was evaluated in two complementary ways.

First, for randomly generated rooted trees, we verified that the output of the algorithm is a valid rainbow $k$-dominating function. In particular, for every vertex assigned the empty set, we checked whether its neighborhood contains all required color sets.

Second, for small instances, we compared the output of the algorithm with an optimal solution obtained by exhaustive search. The brute-force procedure enumerates all possible singleton color assignments and selects a feasible assignment of minimum weight. This verification was applied to trees up to order $n=13$, for which exhaustive search remained computationally feasible.

Together, these two testing procedures verified both feasibility and optimality of the implementation.
\subsubsection*{Runtime evaluation procedure}
Runtime performance was evaluated in two complementary experimental settings.

In the first setting, we investigated the dependence of the running time on the number of vertices $n$, while keeping the number of colors fixed. Trees were generated as uniformly random rooted trees by assigning to each vertex $v_i$, $i>1$, a parent selected uniformly at random from $\{v_1, \dots, v_{i-1}\}$. For each value of $n$, the algorithm was executed multiple times and the median running time was recorded.

In the second setting, we evaluated the dependence of the running time on the number of colors $k$, while keeping the number of vertices fixed. In order to ensure that the influence of $k$ is observable in practice, we generated trees containing a large number of high-degree vertices.

All experiments were repeated multiple times, and median execution times were used to reduce the impact of system noise and transient performance fluctuations. The measurements were performed using a high-resolution steady clock, and each run consisted of a single execution of the algorithm on a freshly generated instance.

\subsubsection*{Hardware environment}
Experiments were conducted on a machine equipped with an Intel Core Ultra 5 125U processor, 16GB of RAM, running Debian 13. The code was compiled using \texttt{g++} version~14.2.0 with optimization flag \texttt{-O3}.

\section{Structural Results for Roman and Rainbow $k$-Domination} \label{sec:structural}
This subsection is devoted to developing structural insights into the nature of rainbow $k$-domination and the constraints imposed by the underlying graph structure. Next, we investigate the structural properties of graphs for which the rainbow 2-domination number is equal to $\frac 32$ times the Roman domination number. Finally, we give computational complexity results comparing the  rainbow 2-domination and Roman domination numbers of bipartite graphs.

\subsection{Lower and Upper Bounds on Rainbow $k$-Domination} 
Since adding an edge does not increase the rainbow $k$-domination  number, one possible approach to find an upper bound of the rainbow $k$-domination number of a graph $G$ is to identify a spanning tree of $G$ and apply the algorithm for trees. In this subsection, we first derive lower bounds for general graphs in terms of graph-theoretic parameters and the number of colors $k$. To complement these results, we also establish an alternative upper bound for general graphs, which may improve the bound obtained from a spanning tree.

Since any rainbow $k$-dominating set is also a $k$-dominating set, for every graph $G$ and every integer $k\geq 1$ it holds that
\[
\gamma_{k}(G)\leq\tilde{\gamma}_{k}(G),
\]
where $\gamma_{k}(G)$ is the $k$-domination number of $G$ \cite{Zerovnik2026}.

Using this property, the following lower bounds are derived from known lower bounds on $k$-domination.

\begin{corollary}
\label{cor:PL1}
Let $G=(V,E)$ be a graph with $n=|V|$ vertices, $m=|E|$ edges and maximum degree $\Delta$, and let $k\ge 2$. Then,
\[
\tilde{\gamma}_{k}(G)\ \ge\ n-\left\lfloor \frac{m}{k}\right\rfloor \qquad \text{and} \qquad \tilde{\gamma}_{k}(G)\ \ge\ \left\lceil \frac{k\,n}{k+\Delta} \right\rceil.
\]
In particular, if $T$ is a tree with $n$ vertices, then
\[
n-\left\lfloor\frac{n-1}k\right\rfloor \leq \tilde{\gamma}_{k}(T)\leq n.
\]
\end{corollary}
\begin{proof}
Fink and Jacobson \cite{Fink1985} proved that for every integer $k\geq1$ 
\[
\gamma_{k}(G)\ \ge\ n- \frac{m}{k} \qquad \text{and} \qquad \gamma_{k}(G)\ \ge\  \frac{k\,n}{k+\Delta}.
\]
Using the general inequality between the rainbow $k$-domination number and the $k$-domination number and accounting for integrality gives the corresponding bounds on rainbow $k$-domination. The lower bound in the special tree case follows from the fact that every tree has $m=n-1$ edges, while the upper bound follows trivially by coloring every vertex.
\end{proof}

The bound involving the maximum degree can be strengthened by recognizing that every set of vertices assigned a particular color must dominate the entire uncolored set independently.

\begin{theorem}
\label{thm:PL1S}
Let $k\geq1$ and let $G$ be a graph of order $n$ and maximum degree $\Delta\geq k$. Then
\[
\tilde{\gamma}_k(G)\geq kq+\min\{k,r\},
\]
where
\[
n=q(k+\Delta)+r,\qquad 0\leq r<k+\Delta.
\]
\end{theorem}
\begin{proof}
Let $f$ be a minimum rainbow $k$-dominating function of $G$. For $i\in\{1,\ldots,k\}$, let $C_i=\{v\in V(G):f(v)=\{i\}\}$, and let $U=\{v\in V(G):f(v)=\emptyset\}$. Then $\tilde{\gamma}_k(G)=\sum_{i=1}^k |C_i|$.

From sets $C_1,C_2,\dots,C_k$, choose one with minimum cardinality. Then $|C_j|\leq\left\lfloor\frac{\tilde{\gamma}_k(G)}{k}\right\rfloor$. Since every vertex of $U$ is adjacent to a vertex of $C_j$ and every vertex of $C_j$ has degree at most $\Delta$, it follows that $|U|\leq \Delta |C_j|\leq \Delta\left\lfloor\frac{\tilde{\gamma}_k(G)}{k}\right\rfloor$. Substituting $|U|=n-\tilde{\gamma}_k(G)$, we obtain
\[
n-\tilde{\gamma}_k(G)\leq\Delta\left\lfloor\frac{\tilde{\gamma}_k(G)}{k}\right\rfloor.
\]

Write $\tilde{\gamma}_k(G)=ak+b$ where $0\leq b<k$. The preceding inequality gives $n\leq a(k+\Delta)+b$. Let $0\leq r<k+\Delta$ and $n=q(k+\Delta)+r$. Suppose that $a<q$. Then 
\[
a(k+\Delta)+b\leq(q-1)(k+\Delta)+k-1<q(k+\Delta)\leq n.
\]
Contradiction. If $a=q$, the inequality $n\leq a(k+\Delta)+b$ implies $b\geq r$ and hence $r<k$. Therefore, in this case $\tilde{\gamma}_k(G)=qk+b\geq qk+r$. If $a>q$, then $\tilde{\gamma}_k(G)=ak+b\geq k(q+1)=kq+k$. Therefore, in both possible cases $\tilde{\gamma}_k(G)\geq kq+\min\{k,r\}$.
\end{proof}

We now present a complementary upper bound.
\begin{theorem}
\label{thm:probabilistic-upper-bound}
Let $k\geq 2$ and let $G$ be a graph of order $n$ and minimum degree $\delta$. Then
\[
\widetilde{\gamma}_k(G)
\leq
n\min\left\{1, \frac{k\bigl(\ln(\delta+k)+1\bigr)}{\delta+k}\right\}.
\]
\end{theorem}
\begin{proof}
Let $p \in [0, 1]$. Independently for each vertex $v$, define $f_0(v)$ by
\[
f_0(v) =
\begin{cases}
    \{i\}, & \text{with probability $p/k$ for each $i \in \{1,\dots,k\}$}, \\
    \emptyset, & \text{with probability $1-p$}.
\end{cases}
\]
Thus, the expected number of initially colored vertices is $pn$. Let $B$ be the set of vertices $v\in V(G)$ satisfying $f_0(v)=\emptyset$ and $\bigcup_{u\in N(v)}f_0(u)\neq\{1,\ldots,k\}$. For a fixed vertex $v$, the probability of a particular color being absent from its open neighborhood is $(1-p/k)^{d(v)}$. The union bound gives
\[
\Pr(v\in B)\leq k(1-p)\left(1-\frac{p}{k}\right)^{d(v)}\leq k(1-p)\left(1-\frac{p}{k}\right)^\delta.
\]

Change the value assigned to the vertices of $B$ from $\emptyset$ to $\{1\}$. The resulting function $f$ is a valid rainbow $k$-dominating function. Its expected weight is
\[
\mathbb E\, w(f) \leq pn+nk(1-p)\left(1-\frac{p}{k}\right)^\delta.
\]
Therefore, for every $p\in[0,1]$, there exists a realization of $f$ whose weight is at most the right-hand side. Taking the minimum over $p$ gives
\[
\widetilde{\gamma}_k(G)\leq n\min_{0\leq p\leq1}\left\{p+k(1-p)\left(1-\frac{p}{k}\right)^\delta\right\}.
\]

To obtain a closed-form bound, we use $(1-p)\left(1-\frac{p}{k}\right)^\delta \leq e^{-p(\delta+k)/k}$. If $p_0=k\ln(\delta+k)/(\delta+k) \leq 1$, taking $p=p_0$ gives
\[
\widetilde{\gamma}_k(G)\leq n\frac{k\bigl(\ln(\delta+k)+1\bigr)}{\delta+k}.
\]
Combining it with the trivial bound $\widetilde{\gamma}_k(G)\leq n$ yields the result.
\end{proof}

\subsection{Rainbow $2$-Domination and Roman Domination}
In this section, we study the mutual relations between the rainbow $2$-domination number and Roman domination number. It is known in the literature~\cite{FUJITA2013806, Chellali2013}, that the Roman domination number $\gamma_R$ of a given graph $G$ is bounded from above by $\frac 32\tilde\gamma_2(G)$. In this section we identify key structural dependencies of graphs for which the upper bound is attained with equality and show that it is NP-hard to decide, for a given graph, whether $\gamma_R(G)=\frac 32\tilde\gamma_2(G)$.

The concept of Roman domination was introduced by Cockayne, Dreyer, Hedetniemi and Hedetniemi in 2004~\cite{cockayne2004roman} as a modification of the classical domination problem. A Roman dominating function on a graph $G=(V,E)$ is a function $f:V\rightarrow\{0,1,2\}$ such that whenever $f(v)=0$ for a vertex $v$, then $v$ is adjacent to some vertex $w$ with $f(w)=2$. For practical notational convenience, we denote $V_i = \{ v \in V(G) : f(v) = i \} \text{ for } i \in \{0,1,2\}$. Therefore, we can write $f=(V_0, V_1, V_2)$. The value of a Roman domination function $f$ is defined to be $|f|=\sum_{v\in V(G)}f(v)$, or equivalently, $|f|=|V_1|+2|V_2|$. The objective is to minimize the total weight $w(f)=\sum_{v\in V}$. The minimum possible weight is called the Roman domination number of $G$ and is denoted by $\gamma_R(G)$. If $|f|=\gamma_R(G)$, then $f$ is said to be a $\gamma_R(G)$-function. Roman domination possesses numerous interesting structural properties, for example if $f$ is a $\gamma_R$-function, then no vertex in $V_1$ is adjacent to a vertex in $V_2$~\cite{cockayne2004roman}.

A graph $G$ is a $(\gamma_R,\frac 32\tilde\gamma_2)$-graph if and only if $\gamma_R(G)=\frac 32 \tilde\gamma_2(G)$. First we present some interesting structural results for $(\gamma_R,\frac 32\tilde\gamma_2)$-graphs.

Let $h$ be a minimum rainbow 2-dominating function of a graph $G$ with $2\gamma_R(G) = 3\tilde\gamma_2(G)$. Then a function $f$ which is obtained from $h$ in such a way that 
\[
f(v)=\begin{cases}
    0, & \mbox{for } h(v)=\emptyset,\\
    1, & \mbox{for } h(v)=\{1\},\\
    2, & \mbox{for } h(v)=\{2\}.
\end{cases}
\]
is a Roman dominating function of weight $\frac 32\tilde\gamma_2(G)$. Since $G$ is $(\gamma_R,\frac 32\tilde\gamma_2)$-graph, $f$ is a minimum Roman dominating function of $G$. 
We call such a function $f$ as \emph{induced Roman dominating function} (obtained from the rainbow 2-dominating function).

Note that if $G$ is a $(\gamma_R,\frac 32\tilde\gamma_2)$-graph, then each minimum rainbow 2-dominating set induces a Roman dominating function and each such function is of minimum weight. On the other hand, if $G$ is a $(\gamma_R,\frac 32\tilde\gamma_2)$-graph, there might exist minimum Roman dominating functions that are not induced by a minimum rainbow 2-dominating function, see Fig.~\ref{fig2}. 

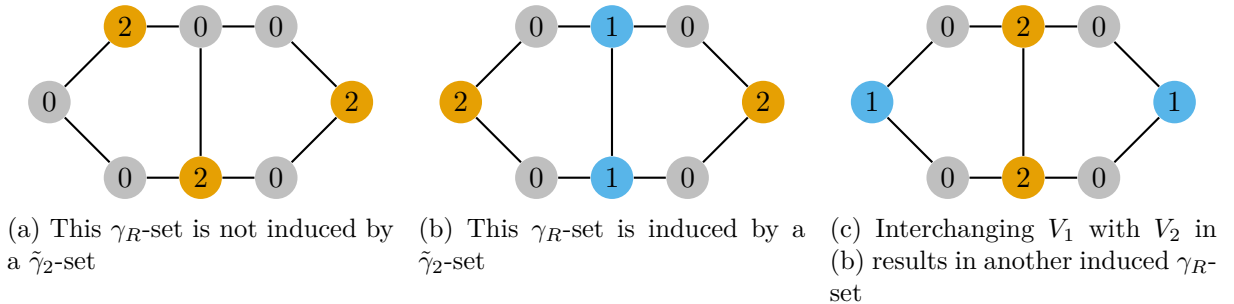
\begin{figure}[h!]
\begin{subfigure}[t]{0.32\textwidth}
\begin{center}
\begin{tikzpicture}[scale=1, rotate=90]
\tikzstyle{vertex0}=[circle,fill=black!25,minimum size=16pt,inner sep=0pt]
\tikzstyle{vertex1}=[circle,fill=oi-lightblue!99,minimum size=16pt,inner sep=0pt]
\tikzstyle{vertex2}=[circle,fill=oi-orange!99,minimum size=16pt,inner sep=0pt]
  \tikzstyle{edge} = [draw, thick,-]
  
\node[vertex0] (x1) at (1,4) {0};
\node[vertex2] (x2) at (2,3) {2};
\node[vertex0] (x3) at (2,2) {0};
\node[vertex0] (x4) at (2,1) {0};
\node[vertex2] (x5) at (1,0) {2};
\node[vertex0] (x6) at (0,1) {0};
\node[vertex2] (x7) at (0,2) {2};
\node[vertex0] (x8) at (0,3) {0};

\path[edge] (x1)--(x2)--(x3)--(x4)--(x5)--(x6)--(x7)--(x8)--(x1);
\path[edge] (x3)--(x7);
\end{tikzpicture}
\caption{This $\gamma_R$-set is not induced by a $\tilde\gamma_2$-set}
\end{center}
\end{subfigure}
\hfill
\begin{subfigure}[t]{0.32\textwidth}
\begin{center}
\begin{tikzpicture}[scale=1, rotate=90]
\tikzstyle{vertex0}=[circle,fill=black!25,minimum size=16pt,inner sep=0pt]
\tikzstyle{vertex1}=[circle,fill=oi-lightblue!99,minimum size=16pt,inner sep=0pt]
\tikzstyle{vertex2}=[circle,fill=oi-orange!99,minimum size=16pt,inner sep=0pt]
  \tikzstyle{edge} = [draw, thick,-]
  
\node[vertex2] (x1) at (1,4) {2};
\node[vertex0] (x2) at (2,3) {0};
\node[vertex1] (x3) at (2,2) {1};
\node[vertex0] (x4) at (2,1) {0};
\node[vertex2] (x5) at (1,0) {2};
\node[vertex0] (x6) at (0,1) {0};
\node[vertex1] (x7) at (0,2) {1};
\node[vertex0] (x8) at (0,3) {0};

\path[edge] (x1)--(x2)--(x3)--(x4)--(x5)--(x6)--(x7)--(x8)--(x1);
\path[edge] (x3)--(x7);
\end{tikzpicture}
\caption{This $\gamma_R$-set is induced by a $\tilde\gamma_2$-set}
\end{center}
\end{subfigure}
\hfill
\begin{subfigure}[t]{0.32\textwidth}
\begin{center}
\begin{tikzpicture}[scale=1, rotate=90]
\tikzstyle{vertex0}=[circle,fill=black!25,minimum size=16pt,inner sep=0pt]
\tikzstyle{vertex1}=[circle,fill=oi-lightblue!99,minimum size=16pt,inner sep=0pt]
\tikzstyle{vertex2}=[circle,fill=oi-orange!99,minimum size=16pt,inner sep=0pt]
  \tikzstyle{edge} = [draw, thick,-]
  
\node[vertex1] (x1) at (1,4) {1};
\node[vertex0] (x2) at (2,3) {0};
\node[vertex2] (x3) at (2,2) {2};
\node[vertex0] (x4) at (2,1) {0};
\node[vertex1] (x5) at (1,0) {1};
\node[vertex0] (x6) at (0,1) {0};
\node[vertex2] (x7) at (0,2) {2};
\node[vertex0] (x8) at (0,3) {0};

\path[edge] (x1)--(x2)--(x3)--(x4)--(x5)--(x6)--(x7)--(x8)--(x1);
\path[edge] (x3)--(x7);
\end{tikzpicture}
\caption{Interchanging $V_1$ with $V_2$ in (b) results in another induced $\gamma_R$-set }
\end{center}
\end{subfigure}
\caption{Minimum Roman dominating functions of a $(\gamma_R,\frac 32\tilde\gamma_2)$-graph, where $\gamma_R(G)=6$ and $\tilde\gamma_2(G)=4$} \label{fig2}
\end{figure}

\begin{theorem}\label{t1}
Let $G$ be a graph with $\gamma_R(G)=\frac 32 \tilde\gamma_2(G)$. Then there exists a minimum Roman dominating function $f=(V_0,V_1,V_2)$ such that

\begin{enumerate}
    \renewcommand{\labelenumi}{(\roman{enumi})}
    \item $|V_1|=|V_2|$. \label{item:equal}
    \item $f'=(V_0,V_2,V_1)$ is also a minimum Roman dominating function. \label{item:swapped}
    \item If $v\in V_0$, then $v$ has a neighbor in $V_1$ and a neighbor in $V_2$. \label{item:neighbors}
\end{enumerate}

\end{theorem}
\begin{proof}
    Let $G$ be a graph with $\gamma_R(G)=\frac 32 \tilde\gamma_2(G)$ and let $h:V(G)\to\{\emptyset, \{1\}, \{2\}\}$ be a minimum rainbow 2-dominating function. Denote by $n_i$, $i=1,2$ the number of vertices $v$ such that $h(v)=\{i\}$. Without loss of generality we may assume that $n_2\leq n_1$. Then $n_2\leq \frac 12 \tilde\gamma_2(G)$ and since each node $z$ with $h(z)=\emptyset$ is adjacent to a node $v$ with $h(v)=\{2\}$,  $\gamma_R(G)\leq 2n_2+n_1=\tilde\gamma_2(G)+n_2\leq \frac 32\tilde\gamma_2(G)$.
    By the assumption, we have equalities in the last inequality chain. This implies that there exists a minimum Roman dominating function $f$ such that \textup{(\ref{item:equal})}, \textup{(\ref{item:swapped})}, and \textup{(\ref{item:neighbors})} follow.
\end{proof}

\begin{theorem}\label{t2}
    If $G$ is a $(\gamma_R,\frac 32\tilde\gamma_2)$-graph, then $\delta(G)\geq 2$.
\end{theorem}
\begin{proof}
    Let $G$ be a $(\gamma_R,\frac 32\tilde\gamma_2)$-graph. Then $G\notin\{K_1, K_2\}$, since for these graphs the Roman domination number and the rainbow 2-domination number are equal.

    Suppose $v_1\in V(G)$ is a vertex of degree~1 and let $v_2$ be the neighbor of $v_1$. Let $f =(V_0, V_1, V_2)$ be a minimum Roman dominating function induced by a rainbow 2-set. Then $f(v_1)\in \{1,2\}$. By Theorem~\ref{t1} we may assume that $f(v_1)=2$. Then $f(v_2)=0$ and there exists another neighbor of $v_2$, say $v_3$, such that $f(v_3)=1$. However in this case a Roman dominating function $f'$ such that $f'(v_1)=f'(v_3)=0$, $f'(v_2)=2$ and $f'(x)=f(x)$ for every other vertex of $G$ has a smaller weight than $f$, which is a contradiction.

\end{proof}

In particular, Theorem~\ref{t2} implies that no tree is a $(\gamma_R, \frac 32 \tilde\gamma_2)$-graph.

\begin{theorem}\label{t3}
    Let $G$ be a $(\gamma_R,\frac 32\tilde\gamma_2)$-graph and let $f=(V_0,V_1,V_2)$ be an induced minimum Roman dominating function. If $uv\in E(G)$ is a bridge, then $|\{u,v\}\cap V_0|\geq 1$. 
\end{theorem}
\begin{proof}
    Let $uv\in E(G)$ be a bridge in a $(\gamma_R,\frac 32\tilde\gamma_2)$-graph, where $f=(V_0,V_1,V_2)$ is an induced minimum Roman dominating function, and suppose 
    $|\{u,v\}\cap V_0|=0$. Note first that if $u\in V_1$ and $v\in V_2$ (or the other way), then $(V_0\cup\{u\}, V_1-\{u\},V_2)$ is a smaller Roman dominating function of $G$, a contradiction. Hence, either $\{u,v\}\subseteq V_1$ or $\{u,v\}\subseteq V_2$.

    Without loss of generality let $\{u,v\}\subseteq V_1$ and denote by $G_u$ and $G_v$ the two connected components of $G-uv$, such that $u\in V(G_u)$ and $v\in V(G_v)$. Then it is easy to see that $\tilde\gamma_2(G_u)+\tilde\gamma_2(G_v)= \tilde\gamma_2(G)$ and $\gamma_R(G_u) + \gamma_R(G_v) = \gamma_R(G)$. Hence, both $G_u$ and $G_v$ are $(\gamma_R, \frac 32 \tilde\gamma_2)$-graphs. Let $f_u$ be a minimum Roman dominating function induced by a rainbow 2-dominating function in $G_u$ such that $f_u(u)=1$ and let $f_v$ be a minimum Roman dominating function induced by a rainbow 2-function in $G_v$ such that $f_v(v)=1$. Then $f_u\cup f_v$ is a minimum Roman dominating function in $G$ of the same weight as $f$. However, since there is an edge connecting a vertex belonging to $V_1$ to a vertex belonging to $V_2$, $f_u\cup f_v$ (and therefore $f$) is not a minimum Roman dominating function of $G$, a contradiction. Therefore, the result follows.
\end{proof}

\begin{theorem}\label{t4}
    Let $G$ be a $(\gamma_R,\frac 32\tilde\gamma_2)$-graph. Then
    \begin{enumerate}
     \renewcommand{\labelenumi}{(\roman{enumi})}
        \item There is no $\tilde\gamma_2(G)$-set such that a vertex belonging to $V_{\{1\}}$ is adjacent to a vertex belonging to $V_{\{2\}}$.
        \item No vertex is adjacent to more than~2 vertices belonging to $V_{\{1\}}$ (or to $V_{\{2\}}$, by symmetry).
        \item If $v\in V_{\{1\}}$, then $v$ is adjacent to at most one other vertex in $V_{\{1\}}$ (analogous result is true for $V_{\{2\}}$). 
    \end{enumerate}
\end{theorem}

\begin{proof}
    Let $G$ be a $(\gamma_R,\frac 32\tilde\gamma_2)$-graph and let $h$ be a minimum rainbow 2-dominating function such that $uv\in E(G)$ and $u\in V_{\{1\}}$, while $v\in V_{\{2\}}$. Consider the Roman dominating function $f=(V_0, V_1, V_2)$ induced by $h$. Since $f$ is of minimum weight, a result published in~\cite{cockayne2004roman} imply that no edge of $G$ joins $V_1$ and $V_2$. For this reason \textit{(i)} is true.

    Similar reasoning and properties of minimum Roman dominating functions proven in~\cite{cockayne2004roman} can be used to justify the results expressed in \textit{(ii)}. and~\textit{(iii)}.  
\end{proof}

\subsubsection*{NP-hardness of deciding $\gamma_R(G)=\frac 32\tilde\gamma_2(G)$}

Notwithstanding the numerous known properties of $(\gamma_R,\frac 32\tilde\gamma_2)$-graphs, determining whether a given graph belongs to this family is NP-hard, even when restricted to bipartite graphs.

The proof is based on presenting a polynomial transformation from the Not-All-Equal 3-Satisfiability (NAE-3SAT) problem, which is a variant of the 3SAT problem. NAE-3SAT is satisfied if and only if the three values in each clause are not all equal to each other, that is, at least one is true and at least one is false. In this proof we use structural results obtained in Theorem~\ref{t4} for $(\gamma_R,\frac 32\tilde\gamma_2)$-graphs. 

\begin{theorem}\label{tNP}
    For a given graph $G$, it is NP-hard to determine whether  $\gamma_R(G) = \frac 32 \tilde\gamma_2(G)$, even when $G$ is bipartite.
\end{theorem}

\begin{proof}
 
Given an instance $E$, that is the set of literals $U = \{u_1, u_2, \dots , u_n\}$ and the set of clauses $C = \{C_1, C_2, \dots , C_m\}$ of NAE-3SAT, we construct a graph $G$ whose order is polynomially bounded in terms of $n$ and $m$, and such that the formula is satisfiable in terms of NAE-3SAT if and only if $\gamma_R(G) = \frac 32 \tilde\gamma_2(G)$.

For each literal $u_i$ construct a copy of the cycle $C_4$, that is $G(u_i)=(u_i,a_i, u'_i, b_i)$ for $i=1,2,\dots ,n$. For each clause  $C_j, j=1, 2,\dots ,m$ there are three nodes, $x_j^1, x_j^2, x_j^3$. All nodes of $G(C_j))$ are adjacent to literal nodes in such a way that in the resulting graph $G$ each of them is of degree~3. 

Additionally, for every clause $C_j$ with literals $l_1, l_2$ and $l_3$, we create the nine edges $xj^1l_t, x_j^2l_t, x_j^3l_t$ for $t=1,2,3$ (see Fig.\ref{fig3} for an example $C_1=(\lnot u_1 \lor u_2\lor u_3)$). It is easily seen that the obtained graph $G$ has $n(G)=4n+3m$ vertices, $m(G)=4n+9m$ edges and is bipartite.

\begin{figure}[h!]
\begin{center}
\begin{tikzpicture}[scale=0.75, rotate=90]
\tikzstyle{vertex0}=[circle,fill=black!25,minimum size=16pt,inner sep=0pt]
  \tikzstyle{edge} = [draw, thick,-]
  
\node[vertex0] (x3) at (5,0) {$x_1^3$};
\node[vertex0] (x2) at (5,4) {$x_1^2$};
\node[vertex0] (x1) at (5,8) {$x_1^1$};

\node[vertex0] (u3) at (9,1) {$u_3$};
\node[vertex0] (u2) at (9,5) {$u_2$};
\node[vertex0] (u1) at (9,9) {$u_1$};

\node[vertex0] (u3') at (9,-1) {$u'_3$};
\node[vertex0] (u2') at (9,3) {$u'_2$};
\node[vertex0] (u1') at (9,7) {$u'_1$};

\node[vertex0] (a3) at (9,0) {$a_3$};
\node[vertex0] (a2) at (9,4) {$a_2$};
\node[vertex0] (a1) at (9,8) {$a_1$};

\node[vertex0] (b3) at (10,0) {$b_3$};
\node[vertex0] (b2) at (10,4) {$b_2$};
\node[vertex0] (b1) at (10,8) {$b_1$};

\foreach \i in {1,2,3} \path[edge] (u\i)--(a\i)--(u\i')--(b\i)--(u\i);
\foreach \w in {u1', u2, u3} \path[edge] (x1)--(\w);
\foreach \w in {u1', u2, u3} \path[edge] (x2)--(\w);
\foreach \w in {u1', u2, u3} \path[edge] (x3)--(\w);

\end{tikzpicture}
\end{center}
\caption{The clause $C_1=(\lnot u_1 \lor u_2\lor u_3)$}\label{fig3}
\end{figure} 

For a minimum rainbow 2-dominating set of $G$ at least two nodes from each $G(u_i)$ should be assigned non-identical color sets in order to dominate $a_i$ and $b_i$. Without loss of generality, we may assume that nodes $u_i$ and $u_i'$ are assigned two different colors, since otherwise. Then, if $E$ is satisfiable in terms of NAE-3SAT, we may assign $\{2\}$ to $u_i$ and $\{1\}$ to $u_i'$, for $i=1,2,\dots, n$ if $u_i$ is true and $\{1\}$ to $u_i$ and $\{2\}$ to $u_i'$ otherwise. Additionally, we assign $\emptyset$ to every other node $y$ of $G$. Since $E$ is satisfiable if and only if in each clause at least one literal is true and at least one is false, each vertex of $x_j^1, x_j^2, x_j^3$, $j=1,2,\dots ,m$ has both colors in its neighborhood. Hence, if $E$ is satisfiable, then $\tilde\gamma_2(G)=2n$. 

If $E$ is not satisfiable, in each assignment of color sets to the literal vertices there are some $x_j^1, x_j^2, x_j^3$, $j=1,2,\dots ,m$ vertices adjacent to just one color set. Hence in this situation $\tilde\gamma_2(G)=2n$, and therefore, $\tilde\gamma_2(G)=2n$ if and only if $E$ is satisfiable in NAE-3SAT.

For the Roman domination observe first, that at most two vertices of each $C_4$ may belong to $V_0$. Therefore, $\gamma_R(G)\geq 3n$. If $E$ is satisfiable in terms of NAE-3SAT, we may construct a Roman dominating function $f$ such that $f(u_i)=2$ and $f(u_i')=1$, $i=1,2,\dots, n$ if $u_i$ is true and $f(u_i)=1$ and $f(u_i')=2$ otherwise, and $f(y)=0$ for every other vertex $y$ of $G$. Since each node of $x_j^1, x_j^2, x_j^3$, $j=1,2,\dots ,m$ is adjacent to a literal vertex $l$ with $f(l)=2$, $f$ is a minimum Roman dominating function of weight $3m$.

As a result, if $E$ has a true/false assignment to the literals in $U$ such that in each clause at least one literal is true and at least one is false, then $\gamma_R(G)=3m=\frac 32 \cdot 2n=\frac 32\tilde\gamma_2(G)$.

Now we prove that if $\gamma_R(G)=\frac 32\tilde\gamma_2(G)$, then $E$ is satisfiable in terms of NAE-3SAT. Let $G$ be a $(\gamma_R,\frac 32\tilde\gamma_2)$-graph, and let $h=(V_\emptyset, V_{\{1\}}, V_{\{2\}})$ be a minimum rainbow 2-dominating function. Without loss of generality let us consider vertices belonging to $V(C_1)$, namely $x_1^1, x_1^2, x_1^3$. Then, by Theorem~\ref{t4}, they cannot belong at the same time to $V_{\{1\}}$ (and to $V_{\{2\}}$, by symmetry). Hence assume, without loss of generality, that $x_1^1, x_1^2\in V_{\{1\}}$, while $x_1^3\in V_{\{2\}}$ and let $l_1$ be a literal vertex adjacent to $x_1^1, x_1^2, x_1^3$. Then, by Theorem~\ref{t4}, $l_1\in V_\emptyset$, $l_1$ has at most one more neighbor in $V_{\{2\}}$ and no other neighbors in $V_{\{1\}}$. However, since $a_1$ and $b_1$ are of degree~2 in $G$, neither $a_1$ nor $b_1$ can belong to $V_\emptyset$. Therefore, we conclude that at least one vertex of $x_1^1, x_1^2, x_1^3$ belongs to $V_\emptyset$.

Assume $x_1^1\in V_\emptyset$. Then at least one of the literal vertices adjacent to $x_1^1$ belongs to $V_{\{1\}}$ and at least one literal vertex belongs to $V_{\{2\}}$. In this situation Theorem~\ref{t4} implies that both $x_1^2$ and $x_1^3$ belong to $V_\emptyset$, too. Hence, $V_{\{1\}}\cup V_{\{2\}}$ is composed solely of vertices of the $G(u_i)$ graphs, $i=1,2,\dots, n$. Since these graphs are disjoint, at least two vertices belonging to each $V(G(u_i))$ are members $V_{\{1\}}\cup V_{\{2\}}$, which in consequence implies that $\tilde\gamma_2(G)\geq 2n$. On the other hand, by Theorem~\ref{t4}, at most two vertices of each $G(u_i)$ belong to $V_{\{1\}}\cup V_{\{2\}}$, and therefore $\tilde\gamma_2(G)\leq 2n$. For these reasons $\tilde\gamma_2(G)= 2n$, and consequently, $E$ is satisfiable in NAE-3SAT.
\end{proof}

 Since each rainbow $k$-dominating set is also a $k$-rainbow dominating set, $\gamma_{rk}(G)\leq\tilde\gamma_k(G)$. Further,
it is proven in~\cite{Wu2010NoteO2, Alvarado2015Relating2D, Chellali2013, FUJITA2013806} that the Roman domination number can be bounded from above by $\frac 32$ times the 2-rainbow domination number. Therefore,
\[
\gamma_R(G)\leq \frac 32 \gamma_{r2}(G)\leq\frac 32\tilde\gamma_2(G).
\]
In contrast to the conjecture given in~\cite{Alvarado2015Relating2D} that graphs satisfying $\gamma_R(G)= \frac 32 \gamma_{r2}(G)$ might admit a polynomial-time recognition algorithm due to their seemingly restrictive structure, we note that the argument used in the proof of Theorem~\ref{tNP} applies, with only minor modifications, to the equality $\gamma_R(G)=\frac 32\gamma_{r2}(G)$. As a consequence, recognizing graphs satisfying $\gamma_R(G)=\frac 32\gamma_{r2}(G)$ is NP-hard as well, and the NP-hardness of recognizing graphs with $\gamma_R(G)=\frac 32\tilde\gamma_{2}(G)$ can be viewed as a stronger structural result, as hardness persists even for this more restricted equality.

\section*{Discussion}

In this paper, we introduced a graph-theoretic framework for modeling the spatial arrangement of facilities in post-disaster shelter camps. Singleton rainbow domination models the placement of distinct essential services, such as sanitation units, kitchens, water points, and schools, while Roman domination provides a complementary framework for services with different capacity levels, illustrated here through Wi-Fi deployment.

Our main algorithmic contribution is the linear-time algorithm for minimum rainbow $k$-domination on trees for fixed $k$. The computational experiments support the theoretical complexity and show good scalability on very large instances. Although realistic shelter-camp graphs are generally not trees, applying the algorithm to a spanning tree of a connected graph yields a feasible solution and hence an efficiently computable upper bound for the original graph. The lower and upper bounds established for general graphs can likewise provide useful estimates of the minimum required number of facilities and the quality of feasible layouts.

Our results on rainbow $2$-domination and Roman domination show that, although the two parameters are closely related, they represent different operational requirements and cannot in general be used interchangeably. In particular, the NP-hardness result demonstrates that recognizing the extremal equality between them remains computationally difficult, motivating the study of restricted graph classes and alternative algorithmic approaches.

\section*{Acknowledgement}
This work has been supported by the European Commission's Horizon Europe Research and Innovation programme through the Marie Skłodowska-Curie Actions Staff Exchanges (MSCA-SE) under Grant Agreement no.101182819 (COVER: (C)ombinatorial (O)ptimization for (V)ersatile Applications to (E)merging u(R)ban Problems).

\bibliographystyle{abbrv}
\bibliography{sample}

\section*{Author contributions statement}

D.G. developed the model, J.R. investigated the problem from both algorithmic and complexity-theoretic perspectives, P.L. conducted the experiments, analysed the results and developed structural bounds. All authors contributed to the preparation of the manuscript draft, reviewed the manuscript, and approved the final version. 

\section*{Additional information}
\textbf{Accession codes}
The complete source code used for the implementation of the algorithm for trees, correctness verification, and runtime benchmarking is publicly available at \url{https://git.pg.edu.pl/pggit_oidc7232510/singleton_rainbow_domination_linear_trees}.

\textbf{Competing interests.} The authors declare that there are no conflicts of interest regarding the publication of this paper.

\end{document}